\documentclass[runningheads,a4paper]{article}

\usepackage{amssymb,amsmath,amsthm}
\usepackage[mathcal]{euscript}
\usepackage{microtype}
\usepackage{array}
\usepackage{graphicx}
\usepackage[margin=1in]{geometry}
\usepackage{hyperref}

\def\beginaalgo{\begin{minipage}{5in}\vspace{0.1cm}\normalfont\begin{tabbing}
       \quad\=\qquad\=\qquad\=\qquad\=\qquad\=\qquad\=\qquad\=\qquad\=\qquad\=\qquad\=\qquad\=\qquad\=\qquad\=\kill}
\def\endaalgo{\end{tabbing}\vspace{0.1cm}\end{minipage}}
\newenvironment{algorithm}
{\begin{tabular}{|l|}\hline\beginaalgo}
{\endaalgo\\\hline\end{tabular}}

\graphicspath{{./fig/}}

\title{$k$-Median clustering under discrete Fr\'{e}chet and Hausdorff distances \thanks{The authors would like to thank Pankaj Agarwal, Kamesh Munagala and anonymous reviewers for helpful discussions and feedback. A shorter version of this paper will appear in SoCG 2020.}}

\author{
Abhinandan Nath\footnote{Part of the work was done when the author was a graduate student at Duke University}\\
    Mentor Graphics, Fremont, USA\\
    \texttt{abnath@mentor.com} \and
Erin Taylor\\
    Duke University, Durham, USA\\
    \texttt{ect15@cs.duke.edu}
}
\date{}

\def\eps{\varepsilon}

\newcommand{\etal}{\textit{et~al.}}
\def\dist{\mathtt{d}}
\newcommand{\E}    {\mathbb{E}}
\newcommand{\A}		{a}
\newcommand{\G}		{b}
\newcommand{\X}		{{\EuScript{X}}}
\newcommand{\C}		{{\EuScript{C}}}
\newcommand{\D}		{t}

\newcommand{\Q}          {h}
\newcommand{\R}          {\theta}
\newcommand{\T}          {{\EuScript{T}}}
\newcommand{\U}          {{\EuScript{U}}}

\newcommand{\ball}      {{\EuScript{B}}}
\newcommand{\reals}    {\mathbb{R}}
\newcommand{\argmin}{\mathop{\mathrm{arg\,min}}}
\newcommand{\ptset}   {\zeta}
\newcommand{\traj}      {\gamma}
\newcommand{\norm}[1]{\left\lVert#1\right\rVert}

\makeatletter
\renewcommand\subparagraph{%
 \@startsection {subparagraph}{5}{\z@ }{3.25ex \@plus 1ex
 \@minus .2ex}{-1em}{\normalfont \normalsize \bfseries }}%
\makeatother

\newtheorem{theorem}{Theorem}
\newtheorem{corollary}[theorem]{Corollary}
\newtheorem{lemma}[theorem]{Lemma}
\theoremstyle{definition}
\newtheorem{definition}[theorem]{Definition}

\begin{document}

\maketitle
\begin{abstract}
We give the first near-linear time  $(1+\eps)$-approximation algorithm for $k$-median clustering of polygonal trajectories under the discrete Fr\'{e}chet distance, and the first polynomial time $(1+\eps)$-approximation algorithm for $k$-median clustering of finite point sets under the Hausdorff distance, provided the cluster centers, ambient dimension, and $k$ are bounded by a constant. The main technique is a general framework for solving clustering problems where the cluster centers are restricted to come from a \emph{simpler} metric space. We precisely characterize conditions on the simpler metric space of the cluster centers that allow faster $(1+\eps)$-approximations for the $k$-median problem. We also show that the $k$-median problem under Hausdorff distance is \textsc{NP-Hard}.
\end{abstract}

\section{Introduction}
\label{sec:intro}
We study the $k$-median problem for an arbitrary metric space $\X = (X,\dist)$, where the cluster centers are restricted to come from a (possibly infinite) subset $C \subseteq X$. We call it the $(k,C)$-median problem. We prove general conditions on the structure of $C$ that allow us to get efficient $(1+\eps)$-approximation algorithms for the $(k,C)$-median problem for any $\eps > 0$. As applications of our framework, we give $(1+\eps)$-approximation algorithms for the metric space defined over polygonal trajectories and finite point sets in $\reals^d$ under the discrete Fr\'{e}chet and Hausdorff distance respectively, where the cluster centers have bounded complexity. For trajectories, our algorithm runs in near-linear time in the number of input points (Theorem~\ref{thm:frechet_clustering}) and is exponentially faster than the previous best algorithm (\cite{BDS19}, Theorem 11). For point sets, ours is the first $(1+\eps)$-approximation algorithm that runs in time polynomial in the number of input points (Theorem~\ref{thm:hausdorff_clustering}) for bounded dimensions and cluster complexity. Our results are summarized in Table~\ref{table:results}. We also show that the $k$-median problem under Hausdorff distance problem is \textsc{NP-Hard}.

\renewcommand{\arraystretch}{1.5}
\begin{table}[h]
\begin{center}
\caption{Our results for $(1+\eps)$-approximate $k$-median for $n$ input trajectories/point sets in $\reals^d$, each having at most $m$ points. Each cluster center can have at most $l$ points. While stating running times,
we assume $k,l$ and $d$ are constants independent of $n,m$, and $\tilde{O}$ hides logarithmic factors in $n,m$.}
\label{table:results}
\begin{tabular}
{ | m{0.4\textwidth} | m{0.2\textwidth} | m{0.2\textwidth} | }
\hline
Metric space & Our result & Previous best\\
\hline
Polygonal traj., discrete Fr\'{e}chet &
$\tilde{O}\left(nm \right)$&
$\tilde{O}\left(n^{dkl+1}m\right)$~\cite{BDS19}\\
\hline
Point sets, Hausdorff &
$nm^{O(dl)}$&
--\\
\hline
\end{tabular}
\end{center}
\end{table}

The $k$-median problem has been very widely studied. We are given a set $P$ of $n$ elements from a metric space. The goal is now to select $k$ \emph{centers} so that the sum of the distances of each point to the nearest cluster center is minimized.
In the simplest setting, $P \subseteq \reals^d$ under the Euclidean metric. In this paper, each individual element of $P$ is itself a collection of points in $\reals^d$, e.g., a curve traced by a moving object, or a point cloud. Since the objective of clustering is to group \emph{similar} objects into the same cluster and to summarize each cluster using its cluster center, it is important that a meaningful distance function is used to compare two input elements. In our work, we look at the widely used discrete Fr\'{e}chet and Hausdorff distances for trajectories and point sets respectively.

Trajectories can model a variety of systems that change with time. As such, trajectory data is being collected at enormous scales. As a first step, clustering is hugely important in understanding and summarizing the data. It involves partitioning a set of trajectories into clusters of similar trajectories, and computing a representative trajectory (a center) per cluster. It can be viewed as a compression scheme for large trajectory datasets, effectively performing non-linear dimension reduction. If the centers have low complexity, this representation can reduce uncertainty and noise found in individual trajectories. Information provided by the set of centers is useful for trajectory analysis applications such as similarity search and anomaly detection~\cite{sung2012trajectory}.  The Hausdorff distance is another widely used shape-based distance~\cite{besse2016,huttenlocher1993comparing}. In shape matching applications, we may want to cluster similar shapes into one group (where a shape is represented by a point cloud).

\begin{figure}
\caption{The red, green and blue trajectories are similar to each other and form a single cluster. If the cluster center trajectory is restricted to have four vertices, it will look like the black trajectory at the bottom. However if the center is unrestricted, it will contain a lot of vertices inside the four noticeable noisy \emph{clumps} of vertices of the input trajectories, thereby overfitting to the input.}
\label{fig:overfitting}
\centering
\includegraphics[width=0.3\textwidth]{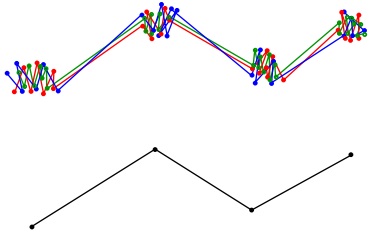}
\end{figure}

The $k$-median problem is hard to solve exactly, even in Euclidean space~\cite{feder1988optimal, MS84}. There is a long line of work on both constant factor and $(1+\eps)$-approximations, with varying running time dependence on $k$ and the ambient dimension (if applicable). Many of these algorithms require the underlying metric space to have bounded doubling dimension (e.g., see~\cite{ABS10}). However, it can be shown that both the discrete Fr\'{e}chet and Hausdorff distances do not have doubling dimension bounded by a constant (Appendix~\ref{app:frecher_hausdorff_doubling_dimension}).
We circumvent this problem by considering cluster centers from a somewhat \emph{simpler} metric space compared to the input metric space. For trajectories and point sets, we restrict the centers to have low complexity, i.e., a bounded number of points. This approach has been used before (\cite{BDGHKL19,BDS19,DKS16}). It has the added benefit of preventing the cluster center from overfitting to the elements of its cluster. This is crucial, since real-life measurements are noisy and error-prone, and without any restrictions the cluster center can inherit noise and high complexity from the input (see Fig.~\ref{fig:overfitting}). As another example, in clustering financial time-series data using Hausdorff distance~\cite{BBDFPP07}, frequent intra-day fluctuations may not be useful in capturing long-term trends, and we want to avoid them by retricting the cluster centers' complexity. However, our work differs from previous approaches in that we precisely characterize general conditions on the simpler metric space for the cluster centers, which leads to faster $(1+\eps)$-approximation algorithms for the $k$-median problem for the discrete Fr\'{e}chet and Hausdorff distances.

\subparagraph*{Problem definition.} A \emph{\textbf{metric space}} $\X = (X,\dist)$ consists of a set $X$ and a distance function $\dist : X \times X \rightarrow \reals_{\ge 0}$ that satisfies the following properties : (i)~$\dist(x,x) = 0$ for all $x \in X$;  (ii)~$\dist(x,y) = \dist(y,x)$ for all $x,y \in X$; and (iii)~$\dist(x,z) \le \dist(x,y) + \dist(y,z)$ for all $x,y,z \in X$.

Given subsets $P, C \subseteq X$, the \emph{\textbf{$(k,C)$-median}} problem is to compute a set $C' \subseteq C$ of $k$ \emph{center} points that minimizes
\begin{align*}
\sum_{p \in P} \dist(p, C'),
\end{align*}
where $\dist(p,C') = \min_{c \in C'} \dist(p,c)$. Here $P$ is finite, but $C$ need not be.

Let $T^l$ (resp. $U^l$) be the set of all trajectories (resp. point sets) in $\reals^d$, where each trajectory (resp. point set) has at most $l$ points. Thus, $T = \bigcup_{l>0} T^l$ and $U = \bigcup_{l>0} U^l$ are the set of all trajectories and finite point sets in $\reals^d$ respectively.
As special cases of the $(k,C)$-median problem, we discuss the $(k,T^l)$-median and the $(k,U^l)$-median problems for the metric spaces $\T = (T,\dist_F)$ and $\U = (U,\dist_H)$ respectively, i.e., each center trajectory or point set can have at most $l$ points. Here, $\dist_F$ and $\dist_H$ denote the discrete Fr\'{e}chet and Hausdorff distances respectively.

\subparagraph*{Challenges and ideas.} As mentioned before, both the discrete Fr\'{e}chet and Hausdorff metrics do not have low doubling dimension, so a number of previous techniques do not directly apply to our setting to yield efficient algorithms. One approach would be to embed these metrics into other metric spaces. Backurs and Sidiropoulos~\cite{BS16} give an embedding of the Hausdorff metric over point sets of size $s$ in $d$-dimensional  Euclidean space, into $l_{\infty}^{s^{O(s+d)}}$ with distortion $s^{O(s+d)}$. However, both the distortion and the resultant dimension are too high for many applications. It is not known if the Fr\'{e}chet distance can be embedded into an $l_p$ space using finite dimension.

We circumvent the problem by restricting the cluster centers to come from a subset $C$ of the original space $X$, namely trajectories and point sets defined by a bounded number of points each. We show that if every metric ball in $C$ can be \emph{covered} by a small number of metric balls of a fixed smaller radius, then we can use a sampling-based algorithm similar to the one in Ackermann \etal~\cite{ABS10}; we make this precise by introducing the notion of \emph{coverability} of $C$. We crucially show that the centers of the balls in the \emph{cover} can be arbitrary, and need not  come from $C$. This is more general than $C$ having bounded doubling dimension. This allows us to approximate the optimal $(1,C)$-median of $P$ using the $(1,C)$-median of a constant sized random sample of $P$, allowing us to use the framework of~\cite{ABS10}.

It is not known how to efficiently compute the optimal $(1,C)$-median under the discrete Fr\'{e}chet and Hausdorff distances. However, we show that all we need is to compute a constant number of candidate centers in time independent of the size of the input, at least one of which is a good  approximation to the optimal $(1,C)$-median of the input. We show how to compute these candidates for a \emph{coverable} set $C$. Then, we can apply the sampling technique of \cite{ABS10} and recursively use this property to find $k$ centers that approximate the cost of the $(k, C)$-median optimal solution.
Although our work heavily relies on the framework of Ackermann et al.~\cite{ABS10}, it is a significant improvement from existing work on clustering under the discrete Fr\'{e}chet distance, and the first such result for clustering under the Hausdorff distance.

\subparagraph*{Previous work.}
Trajectory clustering has a lot of practical applications, e.g., discovering frequent movement patterns in trajectory data. As such, there has been work on trajectory clustering~\cite{GS99,HPL15,XZLZ15}, and possibly computing a representative trajectory for each cluster. Many proposed algorithms and models have no provable performance gaurantees and are experimental in nature.

Driemel \etal~\cite{DKS16} started the rigorous study of clustering trajectories under the continuous Fr\'{e}chet distance under the classic $k$-clustering objectives. However, they only deal with 1D-trajectories. They introduce the $(k,l)$-clustering problem, i.e., clustering trajectories with $k$ cluster centers such that each center can have at most $l$ points. For $1D$-trajectories, they give $(1+\eps)$-approximation algorithms for both the $(k,l)$-center and $(k,l)$-median problem that run in near-linear time for constant $\eps,k,l$. Buchin \etal~\cite{BDGHKL19} study the $(k,l)$-center clustering problem for trajectories in $\reals^d$ under the discrete and continuous Fr\'{e}chet distances, and give both upper and lower bounds. The most closely related work to ours is the one by Buchin, Driemel and Struijs~\cite{BDS19}, where they give algorithms for the $(k,l)$-median problem under discrete Fr\'{e}chet; however their running times are much slower (see Table~\ref{table:results}). They also show that the $1$-median problem under discrete Fr\'{e}chet distance is \textsc{NP-Hard}, and \textsc{W[1]-Hard} in the number of input trajectories.

On the other hand, clustering under Hausdorff distance has received much less attention. There is work on hierarchical clustering of financial time series data using Hausdorff distance~\cite{BBDFPP07}. Chen \etal~\cite{CWLS11} use the DBSCAN algorithm~\cite{EKSX96}  while Qu \etal~\cite{QZC09} use spectral clustering for trajectories and using the Hausdorff distance.  We are not aware of any theoretical analysis for $k$-median clustering of point sets under the Hausdorff distance.

In general metric spaces, a polynomial time $(1+\sqrt{3}+\eps)$-approximation algorithm to the $k$-median problem exists~\cite{LS16}, whereas no polynomial time algorithm can achieve an approximation ratio less than $(1+2/e)$ unless \textsc{NP} $\subseteq$ \textsc{DTIME} $[n ^{O(\log \log n)}]$~\cite{JMS02}. For the Euclidean $k$-median problem in $d$ dimensions, Guruswami and Indyk~\cite{GI03} showed that there is no PTAS for $k$-median if both $k$ and $d$ are part of the input. Arora \etal~\cite{ARR98} gave the first PTAS when $d$ is fixed, whose running time was subsequently improved in~\cite{KR07}. Kumar \etal~\cite{KSS10} gave a $(1+\eps)$-approximate algorithm with runtime $2^{(k/\eps)^{O(1)}}dn$, this was extended by Ackermann \etal~\cite{ABS10} to those metric spaces for which the optimal $1$-median can be approximated using a constant-sized random sample; this holds true for doubling metric spaces. There are coreset-based approaches with running times linear in $n$ and either exponential in $d$ and polynomial in $k$~\cite{HM04,HK07}, or vice versa~\cite{BHI02}. Recently, Cohen-Addad \etal~\cite{CKM19} gave a PTAS for $k$-median in low-dimensional Euclidean and minor-free metrics using local search.

\section{Preliminaries, definitions and an overview}
\label{sec:prelim}
We formally define the discrete Fr\'{e}chet and Hausdorff distances. We also define two properties on $C$ and $\dist$ which allow us to design efficient clustering algorithms. Finally we give an overview of our algorithm.

\subparagraph*{Discrete Fr\'{e}chet and Hausdorff distance.} Consider two finite sets $\ptset_1$ and $\ptset_2$. A \textbf{\emph{correspondence}} $\C$ between $\ptset_1$ and $\ptset_2$ is a subset of $\ptset_1 \times \ptset_2$ such that every element of $\ptset_1$ and $\ptset_2$ appears in at least one pair in $\C$. For $\ptset_1, \ptset_2 \subseteq \reals^d$, the \textbf{\emph{Hausdorff distance}}~\cite{hausdorff78} is defined as
\begin{align*}
\dist_H(\ptset_1, \ptset_2) = \min_{ \C \in \Xi(\ptset_1,\ptset_2)}  \max_{(p,q) \in \C} \norm{p - q},
\end{align*}
where $\norm{.}$ denotes the $l_2$ norm, and $\Xi(\ptset_1,\ptset_2)$ is the set of all possible correspondences between $\ptset_1$ and $\ptset_2$.

A trajectory $\traj$ is a finite sequence of points $\langle p_1, p_2 \ldots\rangle$ in $\reals^d$. A correspodence can be defined for a pair of trajectories $\traj_1 = \langle p_1, p_2,\ldots \rangle$ and $\traj_2 = \langle q_1, q_2, \ldots \rangle$ by treating each trajectory as a point sequence; such a correspondence is said to be \textbf{\emph{monotonic}} if it also respects the ordering of points in the trajectories, i.e., if $(p_{i_1},q_{j_1}), (p_{i_2}, q_{j_2}) \in \C$, then $i_2 \ge  i_1 \Rightarrow j_2 \ge j_1$. The \textbf{\emph{discrete Fr\'{e}chet distance}}~\cite{eiter94} between $\traj_1$ and $\traj_2$ is defined as
\begin{align*}
\dist_F(\traj_1, \traj_2) = \min_{\C \in \Xi_M(\traj_1, \traj_2)} \max_{(p,q) \in \C} \norm{p-q},
\end{align*}
where $\Xi_M(\traj_1, \traj_2)$ is the set of all monotone correspondences between $\traj_1$ and $\traj_2$. Our algorithms cluster in the metric space defined by the discrete Fr\'{e}chet and Hausdorff distance.

\subparagraph*{Strong and weak sampling properties.}
We define two properties that make efficient clustering algorithms possible. These are generalizations of the \emph{strong} and \emph{weak} sampling properties defined by Ackermann \etal (see Theorem 1.1 and Property 4.1 in~\cite{ABS10}). The major difference is that the cluster centers are restricted to a subset $C$ of the metric space $X$. These properties allow fast approximation of the $(1,C)$-median using only a constant sized random sample of the input. We later show how to get an efficient $(k,C)$-median algorithm using the fast $(1,C)$-median algorithm as a subroutine. We denote the optimal $(1,C)$ median of any set $P \subseteq X$ by $c_P$.

\begin{definition}[Strong sampling property]
\label{def:strong_sampling}
Let $0 < \eps, \delta < 1$ be arbitrary. $(X,C,\dist)$ is said to satisfy the strong sampling property for $\eps, \delta$ iff

(i) For any finite $P \subseteq X$, $c_P$ can be computed in time depending only on $|P|$.

(ii) There exists a positive integer $m_{\delta, \eps}$ depending on $\delta, \eps$ such that for any $P \subseteq X$, the optimal $(1,C)$-median $c_S$ of a uniform random multiset $S \subseteq P$ of size $m_{\delta, \eps}$ satisfies
\begin{align*}
\Pr \left[ \sum_{p \in P} \dist(p, c_S) \le (1+\eps) \sum_{p \in P} \dist(p, c_P) \right] \ge 1 - \delta.
\end{align*}

\end{definition}

The strong sampling property characterizes those instances in which the optimal $(1,C)$-median of a constant-sized random sample is a good approximation to the optimal $(1,C)$-median of the whole set. However, in many cases it is impossible to efficiently solve the $(1,C)$-median exactly (e.g., when $C,X = \reals^d$ and $\dist$ is the Euclidean metric). This is also true for the discrete Fr\'{e}chet and Hausdorff distances for polygonal trajectories and finite point sets respectively. The following definition becomes helpful then.

\begin{definition}[Weak sampling property]
\label{def:weak_sampling}
Let $0 < \eps, \delta < 1$ be arbitrary. $(X,C,\dist)$ is said to satisfy the weak sampling property for $\eps, \delta$ iff there exist positive integers $m_{\delta, \eps}$ and $t_{\delta, \eps}$ depending on $\delta, \eps$ such that for any $P \subseteq X$ and a uniform random multiset $S \subseteq X$ of size $m_{\delta, \eps}$, there exists a set $\Gamma(S) \subseteq C$ of size $t_{\delta, \eps}$  that satisfies
\begin{align*}
\Pr \left[ \exists c \in \Gamma(S) \mid \sum_{p \in P} \dist(p, c) \le (1+\eps) \sum_{p \in P} \dist(p, c_P) \right] \ge 1 - \delta.
\end{align*}
Furthermore, $\Gamma(S)$ can be computed in time depending on $\delta, \eps, |S|$ but independent of $|P|$.
\end{definition}
The weak sampling property characterizes those instances in which it is possible to generate a constant number of candidate cluster centers in time independent of the size of the input set, and at least one of which is guaranteed to be a good approximation to the optimal $1$-center of the input set. We later show that the discrete Fr\'{e}chet and Hausdorff distances satisfy the weak sampling property. We use $m_{\delta,\eps}$ to denote the size of the random sample for both the strong and weak sampling properties.

\subparagraph{Algorithm overview.}
\label{sp:alg-overview}
We give an overview of our algorithm, denoted \textsc{Cluster}, in Fig.~\ref{fig:alg_cluster}. It is similar to the algorithm \textsc{Cluster} from~\cite{ABS10}, with the small but crucial difference being that the set of candidate cluster centers $\overline{C}_S$ comes from $C$; this also changes how the candidates are generated. We show that with a careful choice of $C$, our instance satisfies one of the sampling properties (Definitions~\ref{def:strong_sampling},~\ref{def:weak_sampling}), and we can apply the framework of Ackermann et al.~\cite{ABS10}.

The algorithm takes as input the set of points $\overline{P} \subseteq P$ that are yet to be assigned cluster centers, the number of cluster centers $\overline{k}$ still to be computed, and the centers $\overline{C} \subseteq C$ already computed. It returns the final set of cluster centers. To solve the $(k,C)$-median problem for $P$, we call \textsc{Cluster}$(P, k, \{\})$ with values of $\alpha,m_{\delta,\eps},F$ that we will specify later.

Briefly, the algorithm has two phases. In the pruning phase no new centers are added. Rather, the set $N$ containing half of the points of $\overline{P}$ closest to $\overline{C}$ are removed from $\overline{P}$, and the algorithm is called recursively on $\overline{P} \setminus N$. In the sampling phase, new centers are added. The algorithm first samples a uniformly random multiset $S$ of $\overline{P}$ of size $\frac{2}{\alpha} m_{\delta,\eps}$ for some constant $\alpha$ and $m_{\delta,\eps}$ to be defined later. Then for each subset $S' \subset S$ of size $m_{\delta,\eps}$, a set of candidate centers $F(S')$ is generated; the function $F$ varies depending on certain conditions satisfied by $(X,C,\dist)$; in particular $F(S') = \{c_{S'}\}$ for the strong sampling property, and $F(S') = \{\Gamma(S')\}$ for the weak sampling property. Each candidate center is in turn added to $\overline{C}$, and the algorithm is run recursively. Finally, the solution with the lowest cost is returned. \\

\begin{figure}[t]
\centering\small
\begin{algorithm}
\underline{$\textsc{Cluster}(\overline{P}, \overline{k}, \overline{C})$:}\+
\\	\textit{input:} Point set $\overline{P}$, remaining number of centers $\overline{k}$, computed centers $\overline{C}$
\\ if $\overline{k} = 0$: return $\overline{C}$
\\ else:\+
\\ if $\overline{k} \geq |\overline{P}|$: return $\overline{C} \cup \overline{P}$
\\ else:\+
\\ \textit{/* Pruning phase */}
\\ $N \gets$ set of $\frac{1}{2} |\overline{P}|$ minimal points $p \in \overline{P}$ w.r.t $\dist(p, \overline{C})$
\\ $C^* \gets \textsc{Cluster}(\overline{P} \setminus N, \overline{k}, \overline{C})$
\\ \textit{/* Sampling phase */}
\\ $S \gets$ uniform random multisubset of $\overline{P}$ of size $\frac{2}{\alpha} m_{\delta, \eps}$
\\$\overline{C}_S \gets \bigcup_{S' \subset S, |S'| = m_{\delta,\eps}} F(S')$
\\ for all $\overline{c} \in \overline{C}_S$ :\+
\\ $C^{\overline{c}} \gets \textsc{Cluster}(\overline{P}, \overline{k}-1, \overline{C} \cup \{ \overline{c} \})$\-
\\  return $C^{\overline{c}}$ or $C^*$ with the lowest cost
\end{algorithm}
\caption{Algorithm \textsc{Cluster} computes a $(k, C)$-median clustering.}
\label{fig:alg_cluster}
\end{figure}


The rest of the paper is organized as follows. In Section~\ref{sec:sampling}, we show that for $(X,C,\dist)$ satisfying the strong and weak sampling properties, the algorithm \textsc{Cluster}  (using the appropriate $\alpha,m_{\delta,\eps},F$) computes a $(1+\eps)$-approximation to the optimal $(k,C)$-median. In Section~\ref{sec:covering}, we prove sufficient conditions on $C$ for the sampling properties to hold. In Section~\ref{sec:frechet_hausdorff}, we give clustering algorithms for the discrete Fr\'{e}chet and Hausdorff distances using the framework developed in previous sections. In Section~\ref{sec:hardness}, we show that the $k$-median problem for Hausdorff distance is \textsc{NP-Hard}.

\section{Clustering via sampling}
\label{sec:sampling}
We show that the \textsc{Cluster} algorithm (Fig.~\ref{fig:alg_cluster}) computes an approximate $(k,C)$-median for instances satisfying the sampling properties and for appropriate $\alpha, m_{\delta,\eps}$ and $F$, i.e., if we use $F(S') = \{c_{S'}\}$ for the strong sampling property and $F(S') = \Gamma(S')$ for the weak sampling property. The analysis closely follows that of Ackermann \etal~\cite{ABS10}, so we leave detailed proofs to  Appendix~\ref{app:proofs_sec3}.

The \emph{superset sampling lemma} (Appendix~\ref{app:proof_superset_sampling}, Lemma~\ref{lem:superset_sampling}) shows how to draw a uniform random multiset from $P' \subseteq P$ while only knowing $P$ (without explicitly knowing $P'$), provided $P'$ contains a constant fraction of the points of $P$. Using this lemma and the strong and weak sampling properties, we have the following. We leave the proof of the \emph{superset sampling lemma} and the following lemma in Appendices~\ref{app:proof_superset_sampling} and~\ref{app:proof_cluster_algo} respectively.

\begin{lemma}
\label{lem:cluster_algo}
Let $\alpha < \frac{1}{4k}$ be an arbitrary positive constant. Suppose $(X,C,\dist)$ satisfies the strong or weak sampling property (Definitions~\ref{def:strong_sampling},~\ref{def:weak_sampling}) for some $\eps, \delta \in (0,1)$. Given $P \subseteq X$, algorithm \textsc{Cluster} run with input $(P,k,\{\})$ and appropriate $F$ computes a set $\tilde{C} \subseteq C$ of size $k$ such that
\begin{align*}
\Pr \left[ \sum_{p \in P} \dist(p, \tilde{C}) \le (1+8\alpha k^2)(1+\eps) \sum_{p \in P} \dist(p, C^*)\right] \ge \left(\frac{1-\delta}{5}\right)^k,
\end{align*}
where $C^*$ is an optimal solution to the $(k,C)$-median problem for $P$.
\end{lemma}

\subparagraph*{Running time.}


We characterize the running time of \textsc{Cluster} in terms of $\dist$, $C$ and $F$. Let $\Q(C)$ denote the maximum time required to compute $\argmin_{y \in C} \dist(x,y)$ for any $x \in X$, and let $\D(C)$ denote the maximum time required to compute $\dist(x,y)$ for any $x \in X, y \in C$. Finally, let $w(m) = \max_{S \subseteq X, |S| = m} |F(S)|$, and let $f(m)$ be the maximum number of operations needed to compute $F(S)$ for any $S \subseteq X$ of size $m$, where computing $\dist$ between points in $X$ and $C$, and computing the closest point in $C$ to any point in $X$ count as one operation each. The proof is similar to the running time analysis from~\cite{ABS10}, and is given in Appendix~\ref{app:proof_cluster_runtime}. 

\begin{lemma}
\label{lem:cluster_runtime}
Suppose $(X,C,\dist)$ satisfies the strong or weak sampling properties (Definitions~\ref{def:strong_sampling} and~\ref{def:weak_sampling}) for some $\epsilon, \delta \in (0,1)$.  Given $P \subseteq X$ containing $n$ points, algorithm \textsc{Cluster} runs in time
\begin{align*}n \cdot 2^{O\left(k m_{\delta,\eps} \log\left(\frac{1}{\alpha}m_{\delta,\eps}\right)\right)} \cdot \left(w(m_{\delta,\eps}) \cdot f(m_{\delta,\eps})\right)^{O(k)} \cdot (\Q(C) + \D(C)).
\end{align*}
\end{lemma}

By setting $\alpha = \frac{\eps}{8k^2}$, the approximation factor in Lemma~\ref{lem:cluster_algo} becomes $(1+3 \eps)$. Moreover, the error probability can be made arbitrarily small by running the \textsc{Cluster} algorithm $2^{\Theta(k)}$ times and taking the minimum cost solution, without changing the asymptotic running time. We thus get the following. Note that $w(m_{\delta,\eps})$ takes on values $1$ and $t_{\delta,\eps}$ for the strong and weak sampling properties respectively.

\begin{theorem}
\label{thm:k_median_strong_sampling}
Suppose $(X,C,\dist)$ satisfies the strong sampling property (Definition~\ref{def:strong_sampling}) for some $\eps, \delta \in (0,1)$. Further, suppose $c_S$ can be computed in $\A(m_{\delta,\eps})$ operations, where computing $\dist$ between points in $X$ and $C$, and computing the closest point in $C$ to any point in $X$ count as one operation each. Given $P \subseteq X$ having $n$ points and $k \in \mathbb{N}$, with probability $\ge 1-\delta$, a $(1+3\eps)$-approximate solution to the $(k,C)$-median problem for $P$ can be computed in time
\begin{align*}
n \cdot 2^{O\left(k m_{\delta,\eps} \log\left(\frac{k}{\eps}m_{\delta,\eps}\right)\right)} \cdot \A(m_{\delta,\eps})^{O(k)} \cdot (\Q(C) + \D(C)).
\end{align*}
\end{theorem}


\begin{theorem}
\label{thm:k_median_weak_sampling}
Suppose $(X,C,\dist)$ satisfies the weak sampling property (Definition~\ref{def:weak_sampling}) for some $\eps, \delta \in (0,1)$. Further, suppose $\Gamma(S)$ can be computed in  $\G(m_{\delta,\eps})$ operations, where computing $\dist$ between points in $X$ and $C$, and computing the closest point in $C$ to any point in $X$ count as one operation each. Given $P \subseteq X$ having $n$ points and $k \in \mathbb{N}$, with probability $\ge 1-\delta$, a $(1+3\eps)$-approximate solution to the $(k,C)$-median problem for $P$ can be computed in time
\begin{align*}
n \cdot 2^{O\left(k m_{\delta,\eps} \log\left(\frac{k}{\eps}m_{\delta,\eps}\right)\right)} \cdot \left( t_{\delta,\eps} \cdot \G(m_{\delta,\eps})\right)^{O(k)} \cdot (\Q(C) + \D(C)).
\end{align*}
\end{theorem}

\section{Covering metric spaces}
\label{sec:covering}
We specify sufficient conditions on $C$ for the strong and weak sampling properties to hold. These conditions characterize how well can certain subsets of $C$ be \emph{covered} using a small number of sets.

Let $\X = (X,\dist)$ be a metric space. Given $x \in X$, let $\ball_\dist(x,r) = \{x' \in X \mid \dist(x,x') \le r\}$ denote the ball of radius $r$ (under $\dist$) centered at $x$; we will  drop the subscript $\dist$ if it is clear from the context. An \emph{\textbf{r-cover}} of a subset $X' \subseteq X$ for some $r >  0$ is a set $Y \subseteq X$ such that
$
X' \subseteq \bigcup_{y \in Y} \ball(y, r).
$
Note that the elements of $Y$ need not be in $X'$. Also note that if $Y$ is an $r$-cover for $X'$, it is also an $r$-cover for any subset of $X'$.

A subset $Y \subseteq X$ is said to be \emph{\textbf{g-coverable}} for some non-decreasing function $g : \reals_{\ge 0} \rightarrow \reals_{\ge 0}$ iff for all $y \in Y$ and $r > r' > 0$, there exists an $r'$-cover of $\ball(y, r) \cap Y$ of size at most $g(r/r')$. Note that if $Y$ is $g$-coverable then any subset $Y' \subseteq Y$ is also $g$-coverable.

Intuitively, if $C$ has a small cover, then for any metric ball in $C$ there exists a small set of points (not necessarily from $C$), termed the cover, such that the distance from any point in the ball to a point in the cover is smaller than the radius of the ball.

\subparagraph*{Sufficient conditions for the strong sampling property.} The following theorem gives sufficient conditions for the strong sampling property to hold in terms of coverability of $C$. The proof is similar to Lemma~3.4 of \cite{ABS10} but has been adapted to our setting, see Appendix~\ref{app:proof_cover_strong} for more details.

\begin{theorem}
\label{thm:cover_strong}
If $C$ is $g$-coverable, and for any $P \subseteq X$, $c_P$ can be computed in time depending only on $|P|$, then $(X,C,\dist)$ satisfies the strong sampling property (Definition~\ref{def:strong_sampling}) for any $\eps, \delta \in (0,1)$. Here, the constant $m_{\delta, \eps} = m_{\delta, \eps, g}$ also depends on $g$.
\end{theorem}

From Theorems~\ref{thm:k_median_strong_sampling} and~\ref{thm:cover_strong}, we get the following.

\begin{corollary}
\label{cor:k_median_cover_strong}
Suppose $C$ is $g$-coverable and the optimal $(1,C)$-median of any subset of $X$ can be computed in time depending on the size of the subset. Let $\eps, \delta \in (0,1)$. Given $P \subseteq X$ having $n$ points and $k \in \mathbb{N}$, with probability $\ge 1-\delta$, a $(1+3\eps)$-approximate solution to the $(k,C)$-median problem for $P$ can be computed in time
\begin{align*}
n \cdot 2^{O\left(k m_{\delta,\eps,g} \log\left(\frac{k}{\eps}m_{\delta,\eps}\right)\right)} \cdot \A(m_{\delta,\eps,g})^{O(k)} \cdot (\Q(C) + \D(C)),
\end{align*}
where $m_{\delta,\eps,g}$ is a constant depending only on $\eps,\delta,g$, and $\A(m)$ is the number of operations needed to compute the optimal $(1,C)$-median of $m$ points in $X$, where computing $\dist$ between points in $X$ and $C$, and computing the closest point in $C$ to any point in $X$ count as one operation each.
\end{corollary}

\subparagraph*{Sufficient conditions for the weak sampling property.}  We give sufficient conditions for the weak sampling property to hold in terms of coverability of $C$.

For $Y \subseteq X$, let $\R_Y\left(\frac{r}{r'}\right)$ be the number of operations required to compute an $r'$-cover of $\ball(y,r) \cap Y$ for any $y \in Y$ (if such a cover exists) where computing $\dist$ between points in $X$ and $C$, and computing the closest point in $C$ to any point in $X$ count as one operation each (we assume that the number of operations can be expressed in terms of $\frac{r}{r'}$).

The following lemma will be helpful, and states that if $C$ has a small cover and if we have a good estimate of the \emph{cost} of the optimal $(1,C)$-median, then we can construct a small set of points in $C$ such that at least one of them is a good approximation to the optimal $(1,C)$-median. Further, this can be done in time independent of $|P|$. Both of these properties are necessary for the weak sampling property~(Definition~\ref{def:weak_sampling}).

\begin{lemma}
\label{lem:cover_weak}
Let $P \subseteq X$. Suppose $C$ is $g$-coverable. Then given $a, b$ such that $a \le \frac{1}{|P|} \sum_{p \in P} \dist(c_P, p) \le b$, we can compute a set $Q \subseteq C$ of size $O(g\left(\frac{4b}{\eps \delta a}\right))$ such that
\begin{align*}Pr \left[ \exists q \in Q \mid \sum_{p \in P} \dist(p,q) \le (1+\eps) \sum_{p \in P} \dist(p, c_P)\right] \ge 1 - \delta.
\end{align*}
Further, $Q$ can be computed in $O\left(\R_C\left(\frac{4b}{\eps\delta a}\right) + g\left(\frac{4b}{\eps \delta a}\right)\right)$ operations, where computing $\dist$ between points in $X$ and $C$, and computing the closest point in $C$ to any point in $X$ count as one operation each.
\end{lemma}
\begin{proof}
For any $x \in X$, we define $x' = \argmin_{y \in C} \dist(x,y)$.
Consider a point $q \in P$ chosen uniformly at random. By Markov's inequality, $\dist(q,c_P) \le \frac{1}{\delta|P|} \sum_{p \in P} \dist(p,c_P)$ with probability $\ge 1 - \delta$. In such a case,
\begin{align*}
\dist(q', c_P) \le \dist(q',q) + \dist(q,c_p) \le 2 \dist(q,c_P) \le \frac{2}{\delta|P|} \sum_{p \in P} \dist(p,c_P).
\end{align*}
Thus, $c_P \in \ball(q', \frac{2b}{\delta})$ with probability at least $1 - \delta$.

Let $C'$ be an $\left(\frac{\eps a}{2}\right)$-cover of $\ball(q', \frac{2b}{\delta}) \cap C$. Since $C$ is $g$-coverable, $|C'| = g\left(\frac{4b}{\eps\delta a}\right)$. We will argue that $Q = \{q'\} \cup \{ c' \mid c \in C' \}$ is the required solution. Let $x = \argmin_{y \in C'} \dist(y, c_P)$. Since $C'$ is an $\left(\frac{\eps a}{2}\right)$-cover, $\dist(x, c_P) \le \frac{\eps a}{2}$. Also, $\dist(x,x') \le \dist(x,c_P) \le\frac{\eps a}{2}$. Thus, $\dist(x', c_P) \le \dist(x,x') + \dist(x,c_P) \le \eps a$. We then have
\begin{align*}
\sum_{p \in P} \dist(p,x') &\le \sum_{p \in P} \left(\dist(p,c_P) + \dist(x',c_P)\right) \le \left(\sum_{p \in P} \dist(p,c_P)\right) + \eps a |P|\\
&\le (1+\eps) \sum_{p \in P} \dist(p,c_P).
\end{align*}
Computing $C'$ takes $\R_C\left(\frac{4b}{\eps\delta a}\right)$ operations. Computing the output set takes  $|C'|+1  = O\left(g\left(\frac{4b}{\eps\delta a}\right) \right)$ operations.
\end{proof}

The next theorem shows that with a small random sample of $P$, one of two things can happen. Either one of the samples is close to an approximate $(1,C)$-median, or we can approximate the cost of the optimal $(1,C)$-median in time independent of $|P|$. This along with Lemma~\ref{lem:cover_weak} shows that the weak sampling property holds if $C$ is $g$-coverable. The proof is inspired by the proof of Theorem 1 in \cite{KSS05}, and it also shows how to compute $\Gamma$ for the weak sampling property (Definition~\ref{def:weak_sampling}).

\begin{theorem}
\label{thm:cover_weak}
If $C$ is $g$-coverable, then $(X,C,\dist)$ satisfies the weak sampling property (Definition~\ref{def:weak_sampling}) for $0 < \eps < \frac{4}{9}$ and  $\left(1 - \frac{5}{18} \eps\right) < \delta < 1$. Further, the constants $m_{\delta,\eps} = 1+\frac{4}{\eps}$ and $t_{\delta,\eps} =  O\left(g\left(\frac{2048}{\delta_1 \eps^5}\right)\right)$, and the number of operations needed to compute $\Gamma(S)$ is $O\left(\frac{1}{\eps}  + \R_C \left(\frac{2048}{\delta_1\eps^5}\right) + g\left(\frac{2048}{\delta_1\eps^5}\right) \right)$, where  $\delta_1 = \frac{\eps}{2} - \frac{9}{5}(1-\delta)$; and computing $\dist$ between points in $X$ and $C$, and computing the closest point in $C$ to any point in $X$ count as one operation each.
\end{theorem}

\begin{proof}
Let $\eps_1 = \frac{\eps}{4}$ and $\bar{r} = \frac{1}{|P|} \sum_{p \in P} \dist(p, c_P)$. Also, let $x' = \argmin_{y \in C} \dist(x,y)$ for any $x \in X$.

Let $Q \subseteq P$ be a uniform random multiset of size $\frac{1}{\eps_1}$, and $q \in P$ be another point chosen uniformly at random. We will show that $Q \cup \{q\}$ plays the role of $S$ in Definition~\ref{def:weak_sampling}.

Using Markov's inequality and union bound, we have
\begin{align*}
\Pr\left[ \dist(q, c_P) > \frac{\bar{r}}{2 \eps_1^2}\right] < 2 \eps_1^2 \text{  and  }
\Pr\left[ \exists p \in Q \mid \dist(p,c_P) > \frac{\bar{r}}{2\eps_1^2} \right] < \left(\tfrac{1}{\eps_1}\right) 2 \eps_1^2 = 2 \eps_1.
\end{align*}
Thus with probability $\ge 1 - 2\eps_1 - 2\eps_1^2$, $q$ and $Q$ are in $\ball\left(c_P,\frac{\bar{r}}{2\eps_1^2}\right)$. We assume that this event happens. Now, by definition of $q'$, we have $\dist(q, q') \le \dist(q, c_P)$. Hence,
\begin{align*}
\dist(q',c_P) \le \dist(q,q') + \dist(q, c_P) \le 2 \dist(q, c_P) \le \frac{\bar{r}}{\eps_1^2}.
\end{align*}
Let $\ball_1 = \ball\left(c_P, \frac{\bar{r}}{\eps_1^2}\right)$,  $\ball_2 = \ball\left(q', \eps_1 \bar{r}\right)$ and $P' = P \cap \ball_1$. Then, $q' \in \ball_1$ and $Q \subseteq P'$. We consider two cases now.

{\bf Case 1:} $P'$ has at least $2\eps_1 |P'|$ points outside $\ball_2$. For any $p \in Q$, the probability $p$ is outside $\ball_2$ is $2\eps_1$. Thus, with probability at least $2\eps_1$, there exists $p \in Q$ such that $\dist(p,q') \ge \eps_1 \bar{r}$ and hence $\sum_{p \in Q} \dist(p, q') \ge \eps_1 \bar{r}$.
Also, $\dist(p,q') \le \frac{2\bar{r}}{\eps_1^2}$ for any $p \in Q$. Hence, $\sum_{p\in Q} \dist(p,q') \le \frac{2\bar{r}}{\eps_1^3}$.

Let $\delta_1 = 2\eps_1 - \frac{9}{5}(1-\delta)$. We can now use Lemma~\ref{lem:cover_weak} with $a = \frac{\eps_1^3}{2} \sum_{p \in Q}\dist(p,q')$ and $b = \frac{1}{\eps_1}\sum_{p\in Q} \dist(p,q')$ to compute a set $Q_1 \subseteq C$ of $O\left(g\left(\frac{4b}{\eps \delta_1 a}\right)\right) = O\left(g\left(\frac{2048}{\delta_1 \eps^5}\right)\right)$ candidate centers, one of which is a $(1+\eps)$-approximate center with probability at least $1 - \delta_1$. The total probability of getting a good set of candidate centers is $(2\eps_1 - \delta_1)(1-2\eps_1-2\eps_1^2) > 1-\delta$.

{\bf Case 2:} $P'$ has at most $2\eps_1 |P'|$ points putside $\ball_2$. We further consider two cases.

\noindent{\bf Case 2(a):} $\dist(q', c_P) \le 4\eps_1 \bar{r}.$ Then
\begin{align*}
\sum_{p \in P} \dist(p, q') \le \sum_{p \in P} \left( \dist(p, c_p) + \dist(q',c_p)\right) \le (1+4\eps_1)\sum_{p \in P} \dist(p, c_P) \le (1+\eps)\sum_{p \in P} \dist(p, c_P) .
\end{align*}

\noindent{\bf Case 2(b):} $\dist(q', c_P) > 4\eps_1 \bar{r}.$ Suppose we assign all points from $c_P$ to $q'$. By an averaging argument, we have $|P'| \ge (1-\eps_1^2)|P|$. Then, the number of points of $P$ that are outside $\ball_2$ is at most
\begin{align*}
|P \setminus P'| + 2\eps_1 |P'| =& |P| - (1-2\eps_1)|P'| \\
\le& |P| - (1-2\eps_1)(1-\eps_1^2)|P|\\
\le& (\eps_1^2 + 2\eps_1(1-\eps_1^2))|P|.
\end{align*}
Thus, $|P \cap \ball_2| \ge (1 - \eps_1^2 - 2\eps_1(1-\eps_1^2))|P|$.
Now, for $p \in P \cap \ball_2$, the decrease in cost on switching from $c_P$ to $q'$ is at least
\begin{align*}
\dist(p,c_P) - \dist(p,q') \ge \dist(p, c_P) - \eps_1 \bar{r} \ge \dist(q', c_P) - 2 \eps_1 \bar{r}.
\end{align*}
For $p \in P \setminus \ball_2$, the increase in cost on switching from $c_P$ to $q'$ is at most
\begin{align*}
\dist(p, q') - \dist(p, c_P) \le \dist(q', c_P).
\end{align*}
The overall decrease in cost is
\begin{align*}
|P \cap \ball_2| (\dist(q',c_P) - 2\eps_1 \bar{r}) - |P \setminus \ball_2| \dist(q',c_P) > 0
\end{align*}
for our choice of $\eps_1$ and $\dist(q',c_P) > 4 \eps_1 \bar{r}$. But $c_P$ is the optimal $(1,C)$-median of $P$, a contradiction. Hence case 2(b) cannot occur.

From the two cases above, we can see that $Q_1 \cup \{q'\}$ plays the role of $\Gamma(S)$ in Definition~\ref{def:weak_sampling}.

Sampling $Q, q$ takes time $O(\frac{1}{\eps})$. Computing $\sum_{p \in P} \dist(p,q')$ involves $O(\frac{1}{\eps})$ computations of $\dist$ between a pair of points from $X$, at least one of which comes from $C$. Computing the set of candidates takes  $O\left(\R_C\left(\frac{2048}{\delta_1\eps^5}\right) +  g\left(\frac{2048}{\delta_1\eps^5}\right)\right)$ operations. Thus, total number of operations needed is $O\left(\frac{1}{\eps} + \R_C\left(\frac{2048}{\delta_1\eps^5}\right) +  g\left(\frac{2048}{\delta_1\eps^5}\right)\right)$. Moreover, $m_{\delta,\eps} = |Q \cup \{q\}| = 1 + \frac{1}{\eps_1} = 1 + \frac{4}{\eps}$, and $t_{\delta,\eps} = |Q_1 \cup \{q'\}| = O\left(g\left(\frac{2048}{\delta_1 \eps^5}\right)\right)$.
\end{proof}

From Theorems~\ref{thm:k_median_weak_sampling} and~\ref{thm:cover_weak}, we get the following.

\begin{corollary}
\label{cor:k_median_cover_weak}
Suppose $C$ is $g$-coverable. Let $\eps \in (0,\frac{4}{9}), \delta \in (1-\frac{5}{18}\eps,1)$. Given $P \subseteq X$ of size $n$ and $k \in \mathbb{N}$, with probability $\ge 1-\delta$, a $(1+3\eps)$-approximate solution to the $(k,C)$-median problem for $P$ can be computed in time
\begin{align*}
n \cdot 2^{O\left(\frac{k}{\eps} \log\left(\frac{k}{\eps}\right)\right)} \cdot \left(\G(m_{\delta,\eps})\cdot t_{\delta,\eps}\right)^{O(k)} \cdot (\Q(C) + \D(C)),
\end{align*}
where $\G(m_{\delta,\eps}) = O\left(\frac{1}{\eps} + \R_C\left(\frac{2048}{\delta_1\eps^5}\right) +  g\left(\frac{2048}{\delta_1\eps^5}\right)\right)$, $t_{\delta,\eps} = O\left(g\left(\frac{2048}{\delta_1\eps^5}\right)\right)$, and $\delta_1 = \frac{\eps}{2} - \frac{9}{5}(1-\delta)$.
\end{corollary}

\section{Clustering discrete Fr\'{e}chet and Hausdorff distances}
\label{sec:frechet_hausdorff}
In this section, we show how the results from the previous section can be used to cluster trajectories and point sets under the discrete Fr\'{e}chet and Hausdorff distances respectively.

\subparagraph*{Clustering under discrete Fr\'{e}chet distance.} Recall that $T^l$ is the set of all trajectories in $\reals^d$ having at most $l$ points each; thus $T = \bigcup_{l>0} T^l$ is the set of all trajectories in $\reals^d$. Given $\T = (T,\dist_F)$ and trajectories $P \subseteq T$, the $(k,l)$-median problem~\cite{DKS16,BDS19} is equivalent to the $(k,T^l)$-median problem in our setting, i.e., the center trajectories contain at most $l$ points. We show that $T^l$ is $g$-coverable for some $g$ that depends on $l$.

\begin{lemma}
\label{lem:frechet_cover}
$T^l$ is $g$-coverable under $\dist_F$ for $g(x) = l^{2l} \cdot x^{O(dl)}$. Further, an $r'$-cover of $\ball_{\dist_F}(\traj, r) \cap T^l$ for $r > r' > 0$ and $\traj \in T^l$ can be computed in $l^{2l} \cdot \left(\frac{r}{r'}\right)^{O(dl)}$ time.
\end{lemma}
\begin{proof}
Let $r > r' > 0$ be arbitrary. Let $\traj = \langle p_1, \ldots, p_{l'}\rangle \in T^l$ for some $l' \le l$. Since the Euclidean metric in $\reals^d$ has doubling dimension $O(d)$, for any $p \in \reals^d$ there exist $\left(\frac{r}{r'}\right)^{O(d)}$ points in the Euclidean ball $\ball_E(p,r)$ centered at $p$ such that any point in $\ball_E(p,r)$ is at most $r'$ distance away from one of these points; denote these points by $B_p(r,r')$. Consider the set of points $\bigcup_{i=1}^{l'} B_{p_i}(r,r')$; this set has cardinality $l' \cdot \left( \frac{r}{r'}\right)^{O(d)}$.

Next, consider the set of all trajectories $T'$ defined by at most $2l$ points from $\bigcup_{i=1}^{l'} B_{p_i}(r,r')$ and containing at least one point from $B_{p_i}(r,r')$ for every $i$; further these points respect the ordering of the sets that they belong to, i.e., if $p \in B_{p_i}(r,r')$, $q \in B_{p_j}(r,r')$, and $i < j$, then $p$ appears before $q$ in the trajectory (for points coming from the same set $B_{p_i}(r,r')$ all possible orderings are considered). Note that $|T'| \le \left(l' \cdot \left(\frac{r}{r'}\right)^{O(d)}\right)^{2l}$. Further, $\dist_F(\traj, \traj') \le r$ for all $\traj' \in T'$.

We will show that for any $\traj'' \in \ball_{\dist_F}(\traj, r) \cap T^l$, there exists $\traj' \in T'$ such that $\dist_F(\traj', \traj'') \le r'$. Thus $T'$ is the desired cover, completing the first part of our proof. Let $\traj'' = \langle q_1, \ldots, q_{l''}\rangle$ for some $l'' \le l$. By definition of $\dist_F$ and the fact that $\dist_F(\traj, \traj'') \le r$, each $q_i$ has a corresponding sequence of points $\langle p_{j_i}, p_{j_i+1}, \ldots, p_{j'_i} \rangle$ each of which is at most $r$ distance away from $q_i$. Moreover, for all $1 \le i < l''$ we have $j'_i \le j_{i+1} \le j'_i + 1$, and $j_1 = 1, j'_{l''} = l'$.

For each $q_i$ and $j \in \{j_i, j_i+1, \ldots, j'_i\}$, let $u_j$ denote the point in $B_{p_j} (r,r')$ that is closest to $q_i$. Note that $q_i \in \ball_E(p_j,r)$ and $\norm{q_i - u_j} \le r'$. Consider the sequence of points $\traj(q_i) = \langle u_{j_i}, u_{j_i+1},\ldots,u_{j'_i} \rangle$. Then, $\dist_F(\langle q_i \rangle, \traj(q_i)) \le r'$. Let $\traj'$ be the trajectory obtained by concatenating $\traj(q_1), \ldots, \traj(q_{l''})$ in order. Then we get $\dist_F(\traj', \traj'') \le r'$. Further, by construction $\traj'' \in T'$.

As far as running time is concerned, computing the set $\bigcup_{i=1}^{l'} B_{p_i}(r,r')$ takes time $l' \cdot \left( \frac{r}{r'}\right)^{O(d)}$. From this set, computing $T'$ takes time $l^{2l} \cdot \left(\frac{r}{r'}\right)^{O(dl)}$.
\end{proof}

For a trajectory having $m$ points, computing $\dist_F$ to any trajectory in $T^l$ takes time $O(ml)$ using the standard dynamic programming algorithm, whereas computing the closest trajectory in $T^l$ under $\dist_F$ takes time $O\left(l m \log m \log \left( \frac{m}{l}\right)\right)$ time (see~\cite{BJWYZ08}, Theorem 3). This, along with Corollary~\ref{cor:k_median_cover_weak} and Lemma~\ref{lem:frechet_cover} give the following.

\begin{theorem}
\label{thm:frechet_clustering}
Let $\eps \in (0,\frac{4}{9}), \delta \in (1-\frac{5}{18}\eps,1)$. Given a set of $n$ trajectories $P \subseteq T$ each having at most $m$ points and $k \in \mathbb{N}$, with probability $\ge 1-\delta$, the algorithm \textsc{Cluster} (Fig.~\ref{fig:alg_cluster}) computes a $(1+3\eps)$-approximate solution to the $(k, T^l)$-median problem for $P$ under the discrete Fr\'{e}chet distance in time
\begin{align*}
nm\log m \log \left(\frac{m}{l}\right) 
\cdot 2^{O\left(\frac{k}{\eps} \log\left(\frac{k}{\eps}\right)\right)} 
\cdot \left( \frac{l}{\delta_1 \eps} \right)^{O(kdl)} ,
\end{align*} 
where $\delta_1 = \frac{\eps}{2} - \frac{9}{5}(1-\delta)$. 
\end{theorem}

\subparagraph*{Clustering under Hausdorff distance.} Recall that $U^l$ is the set of all point sets in $\reals^d$ containing at most $l$ points each. Thus, $U = \bigcup_{l>0} U^l$ is the set of all finite point sets of $\reals^d$. Given $\U = (U,\dist_H)$ and subsets $P \subseteq U$, we show how to approximately solve the $(k, U^l)$-clustering problem for $P$ and $k,l > 0$.

\begin{lemma}
\label{lem:hausdorff_cover}
$U^l$ is $g$-coverable under $\dist_H$ for $g(x) = l^l \cdot x^{O(dl)}$. Further, an $r'$-cover of $\ball_{\dist_H}(\ptset, r) \cap U^l$ for $r > r' > 0$ and $\ptset \in U^l$ can be computed in $l^l \cdot \left( \frac{r}{r'} \right)^{O(dl)}$ time.
\end{lemma}
\begin{proof}
Let $r > r' > 0$ be arbitrary, and let $\ptset = \{p_1, p_2, \ldots, p_{l'} \}$ for some $l' \le l$. Since the Euclidean metric has doubling dimension $O(d)$, there exist $\left(\frac{r}{r'}\right)^{O(d)}$ points in the Euclidean ball $\ball_E(p,r)$ such that any point in $\ball_E(p,r)$ is at most $r'$ distance away from one of these points; denote these points by $B_p(r,r')$. 

Consider all subsets of size $l$ of the set $\bigcup_{i=1}^{l'} B_{p_i}(r,r')$, denote this set by $U'$. Note that $|U'| = \left( l' \cdot \left( \frac{r}{r'} \right)^{O(d)}\right)^l$.

Next, consider any point set $\ptset' \in \ball_{\dist_H}(\ptset, r) \cap U^l$. By definition of Hausdorff distance, the points of $\ptset'$ (there are at most $l$ of them) must lie in $\bigcup_{i=1}^{l'} \ball_E(p_i, r)$. Thus, for each $p \in \ptset'$, there is some $i \in \{1, \ldots, l'\}$ such that $\norm{p-q} \le r'$ for some $q \in B_{p_i}(r,r')$. Thus, there exists $\ptset'' \in U'$ such that $\dist_H(\ptset', \ptset'') \le r'$. Then, $U'$ is the desired cover, completing the first part of our proof.

Computing $\bigcup_{i=1}^{l'} B_{p_i}(r,r')$ takes time $l' \cdot \left(\frac{r}{r'}\right)^{O(d)}$. Computing $U'$ from it takes time $\left( l' \cdot \left( \frac{r}{r'} \right)^{O(d)}\right)^l$.
\end{proof}

Computing $\dist_H$ between two point sets of size $m_1$ and $m_2$ in $\reals^d$ takes time $O(m_1m_2)$. Given $\ptset \in U$ of size $m$, computing the closest point to it in $U^l$ (under $\dist_H$) boils down to finding $l$ disks in $\reals^d$ of minimum radius such that all the points of $\ptset$ lie inside the union of these disks; the centers of these disks give the desired set in $U^l$. This is the $l$-center problem in $\reals^d$ for the Euclidean metric. Solving this problem is \textsc{NP-Hard} when $l$ is part of the input. Note that the total number of subsets of $\ptset$ induced by disks in $\reals^d$ is $m^{O(d)}$. By looking at $l$ such subsets at a time, we can pick the one that covers $\ptset$ and minimizes the radius of the largest disk; this takes total time $m^{O(dl)}$. This along with Corollary~\ref{cor:k_median_cover_weak} and Lemma~\ref{lem:hausdorff_cover} give the following.

\begin{theorem}
\label{thm:hausdorff_clustering}
Let $\eps \in (0,\frac{4}{9}), \delta \in (1-\frac{5}{18}\eps,1)$. Given a set of $n$ point sets $P \subseteq U$ each having at most $m$ points and $k \in \mathbb{N}$, with probability $\ge 1-\delta$, the algorithm \textsc{Cluster} (Fig.~\ref{fig:alg_cluster}) computes a $(1+3\eps)$-approximate solution to the $(k, U^l)$-median problem for $P$ under the Hausdorff distance can be computed in time
\begin{align*}
nm^{O(dl)}
\cdot 2^{O\left(\frac{k}{\eps} \log\left(\frac{k}{\eps}\right)\right)} 
\cdot \left( \frac{l}{\delta_1 \eps} \right)^{O(kdl)} ,
\end{align*} 
where $\delta_1 = \frac{\eps}{2} - \frac{9}{5}(1-\delta)$. 
\end{theorem}

\subparagraph*{Remark.}
We believe the running time of Theorem~\ref{thm:hausdorff_clustering} can be improved to be similar to that of Theorem~\ref{thm:frechet_clustering}, since a $\left(1+O(\eps)\right)$-approximate nearest neighbor (instead of the actual nearest neighbor) in $C$ for any $x \in X$ should suffice. For the Hausdorff case, this involves using a fast approximation for the $l$-center problem~\cite{agarwal2002exact}.
\section{Hardness of $k$-median clustering under Hausdorff distance}
\label{sec:hardness}

We prove the following hardness result for $k$-median clustering under the Hausdorff distance.

\begin{theorem}
The $k$-median clustering problem for finite point sets under the Hausdorff distance is \textsc{NP-Hard}.
\end{theorem}
\begin{proof}
We reduce the Euclidean $k$-median problem, which is known to be \textsc{NP-Hard}~\cite{MS84}. The reduction is fairly straightforward -- for each input point $p$ of an instance of the Euclidean $k$-median problem, we have a singleton set $\{p\}$ as input to the Haudorff $k$-median problem. Any solution to the Euclidean $k$-median problem is also a solution to the Hausdorff $k$-median problem of the same cost -- we replace cluster center $c$ in the Euclidean version by the cluster center $\{c\}$ for the Haudorff version, and for each $p$ assigned to $c$, we assign $\{p\}$ to $\{c\}$.

On the other hand, consider a solution to the instance of Hausdorff $k$-median problem. In particular, let $S$ be a cluster center that is assigned the sets $\{ \{p_1\}, \ldots, \{p_n\}\}$. The cost of this single cluster is 
\begin{align*}
\sum_{i=1}^n \dist_H (\{p_i\}, S) = \sum_{i=1}^n \max_{s \in S} \norm{s-p_i} , 
\end{align*}
by the definition of $\dist_H$. Thus, replacing $S$ by a singleton set $\{s\}$ for any $s \in S$ does not increase the cost of clustering. Hence we can assume that all cluster centers are singleton sets. We can then construct a solution for the Euclidean $k$-median problem by assigning $p$ to $s$, where $\{s\}$ is the cluster center that $\{p\}$ was assigned to in the Hausdorff clustering solution. This does not increase the cost of the clustering as well.
\end{proof}
\section{Conclusion}
\label{sec:conclusion}

We have given a framework for clustering where the cluster centers are restricted to belong to a simpler metric space. We characterized general conditions on this simpler space that allow us to obtain efficient $(1+\eps)$-approximation algorithms for the $k$-median problem. As special cases, we gave efficient algorithms for clustering trajectories and point sets under the discrete Fr\'{e}chet and Hausdorff distances respectively.

We believe the general framework can be extended to other metric spaces as well. The next step would be to see if it can be applied to the continuous Fr\'{e}chet distance, and to non-metric distance measures such as dynamic time warping. It would be interesting to provide other characterizations on the metric space for the cluster centers (as alternatives to the notion of covering discussed in this paper) that are amenable to efficient clustering algorithms.

\bibliography{ref}

\appendix

\section{Doubling dimension of Fr\'{e}chet and Hausdorff distance}
\label{app:frecher_hausdorff_doubling_dimension}
We show that the discrete Fr\'{e}chet distance does not have a doubling dimension bounded by a constant. The Hausdorff case can be shown similarly.

Suppose $\T = (T, \dist_F)$ has a constant doubling dimension $D$. Then by definition, for any trajectory $\traj \in T$ and any $r > 0$, there exist $2^D$ trajectories $T' \subset T$ such that any trajectory $\traj' \in \ball_{\dist_F}(\traj, r)$ is at most at distance $\frac{r}{2}$ from a trajectory in $T'$. By the pigeon hole principle, any set of $2^D + 1$ or more trajectories in $\ball_{\dist_F}(\traj,r)$ will have at least two trajectories that have the same closest trajectory in $T'$, and are therefore at most $r$ distance apart (by the triangle inequality).

Consider a trajectory $\traj''$ consisting of a sequence of points $\langle p_1, p_2, \ldots, p_m \rangle$ in a straight line in $\reals^d$, with a distance of at least $3r$ between every consecutive pair of points, for some $m,r > 0$. For each $p_i$, let $P_i = \{a_i, b_i\}$ be such that $\norm{a_i-b_i} > r$ and $\norm{a_i-p_i}, \norm{b_i - p_i} \le r$ (e.g., $a_i$ and $b_i$ can be the diametrically opposite points of the Euclidean ball of radius $r$ centered at $p_i$). Consider the set of trajectories $T'' = \{ \langle q_1, q_2, \ldots, q_m \rangle \mid q_i \in P_i\}$. Then by construction, $T'' \subset \ball_{\dist_F}(\traj'', r)$, and $\dist_F(\traj_1, \traj_2) > r$ for any $\traj_1,\traj_2 \in T''$. Further, $|T''| = 2^m \ge 2^D + 1$ for sufficiently large $m$, contradicting the fact that $\T=(T, \dist_F)$ has doubling dimension $D$.

\section{Proofs of Section~\ref{sec:sampling}}
\label{app:proofs_sec3}

\subsection{Proof of the superset sampling lemma}
\label{app:proof_superset_sampling}
We state the {\it superset sampling lemma} for our setting. As in Section~\ref{sec:sampling}, we define  $F(S')$ to be $\{ c_{S'}\}$ or $\Gamma(S')$ for instances satisfying the strong or weak sampling property respectively. The lemma states that we can draw a uniform sample from $P' \subseteq P$ without explicitly knowing $P'$, provided $P'$ contains a sufficient fraction of points of $P$. In our algorithm in Fig.~\ref{fig:alg_cluster}, we use this property in the \emph{sampling phase} in order to sample points from a single cluster, without explicitly knowing the points in the said cluster.

\begin{lemma}
\label{lem:superset_sampling}
Suppose $0 < \alpha < \frac{1}{4}$. Suppose $(X, C, \dist)$ satisfies either the strong or weak sampling property. Let $P \subseteq X$ of size $n$ and $P' \subseteq P$ with $|P'| \geq \alpha n$. Let $S \subseteq P$ be a uniform sample multiset of size at least $\frac{2}{\alpha} m_{\delta, \eps}$. Then with probability at least  $\frac{1-\delta}{5}$, there exists a subset $S' \subseteq S$ with $|S'| = m_{\delta, \eps}$ satisfying the following.
\begin{align*}
\exists c \in F(S') \text{ s.t. } \sum_{p \in P} \dist(p,c) \le (1+\eps) \sum_{p \in P} \dist(p,c_P).
\end{align*}
Here, $F(S')$ is either $\{c_{S'}\}$ or $\Gamma(S')$ for $(X,C,\dist)$ satisfying the strong or weak sampling property respectively.
\end{lemma}

\begin{proof}
 This proof is similar to that of Lemma 2.1 in~\cite{ABS10}. Define $Y$ to be a random variable denoting the number of points from $P'$ contained in $S$. Note that $\E[Y] \ge 2m_{\delta, \eps}$, since $S$ is sampled with replacement. By applying a Chernoff bound, we obtain
\begin{align*}
 \Pr[Y < m_{\delta, \eps}] \leq \Pr \left[Y < \frac{\E[Y]}{2} \right] \le e^{-m_{\delta, \eps}/4} \le e^{-1/4} < \frac{4}{5},
\end{align*}
since $m_{\delta, \eps} \geq 1$. Thus, with probability at least $\frac{1}{5}$, $S$ containts at least $m_{\delta,\eps}$ points of $P'$. This along with the strong or weak sampling property finishes the proof.
\end{proof}

%

\subsection{Proof of Lemma~\ref{lem:cluster_algo}}
\label{app:proof_cluster_algo}
The following is similar to the proof of Theorem 2.2~\cite{ABS10}, as the correctness does not depend on how candidate centers are found. We include it for completeness.
\begin{proof}
We prove the lemma for $k = 2$; it generalizes in a straightforward manner for $k > 2$.
We assume $n = |P|$ is a power of $2$ for simplicity.
Suppose $P_1, P_2$ are the clusters with centers $C^* = \{ c^*_1, c^*_2 \}$ corresponding to the optimal $(2,C)$-median for $P$.
Assume $|P_1| \geq \frac{1}{2} P_2$.
Let $opt_1(Q)$ and $opt_2(Q)$ denote the \emph{value} of the optimal $(1,C)$-median and $(2,C)$-median respectively.

By Lemma~\ref{lem:superset_sampling}, during the sampling phase of the initial call to \textsc{Cluster}, $\overline{C_S}$ contains a $\overline{c_1}$ such that $\sum_{p \in P_1} \dist(p, \overline{c_1}) \leq (1+\eps)opt_1$.
We consider two cases: when the algorithm selects $\overline{c_1}$ and recurses with $(\overline{P}, 1, \{ \overline{c_1} \})$ where $\overline{P}$ contains a suitably large number of points from $P_2$, and when this does not occur.

{\bf Case 1:} Suppose there exists a call with $(\overline{P}, 1, \{ \overline{c_1} \})$ such that $|P_2 \cap \overline{P}| \geq \alpha |\overline{P}|$.
Then, by Lemma~\ref{lem:superset_sampling}, during this call $\overline{C_S}$ contains a $\overline{c_2}$ such that $\sum_{p \in P_2 \cap \overline{P}} \dist(p, \overline{c_2}) \leq (1+\eps)opt_1(P_2 \cap \overline{P})$.
In this case, we upper bound the cost of $\overline{C} = \{\overline{c_1}, \overline{c_2} \}$.
Let $N = P \setminus \overline{P}$ be the points removed by the pruning phase between the sampling of $\overline{c_1}$ and $\overline{c_2}$. $P_1, P_2 \cap N$, and $P_2 \cap \overline{P}$ form a partition of $P$. We have that:

 \begin{equation}
     \label{eq:alg-cost}
     \sum_{p \in P} \dist(p, \overline{C}) \leq \sum_{p \in P_1} \dist(p, \overline{c_1}) + \sum_{p \in P_2 \cap N} \dist(p, \overline{c_1}) + \sum_{p \in P_2 \cap \overline{P}} \dist(p, \overline{c_2}).
 \end{equation}

By the definition of $\overline{c_1}$, we can bound the first term of Equation~\ref{eq:alg-cost}:

\[ \sum_{p \in P_1} \dist(p, \overline{c_1}) \leq (1+\eps) \sum_{p \in P_1} \dist(p, c^*_1). \]

Next, we bound the last term of Equation~\ref{eq:alg-cost}. By the selection of $\overline{c_2}$:
\[  \sum_{p \in P_2 \cap \overline{P}} \dist(p, \overline{c_2}) \leq (1+ \eps) \sum_{p \in P_2 \cap \overline{P}} \dist(p, c_{P_2 \cap \overline{P}}) \leq (1+\eps) \sum_{p \in P_2} d(p, c^*_2). \]

Now, we only need to bound the middle term of Equation~\ref{eq:alg-cost}. We first assume $N \neq \emptyset$, otherwise we are done. Suppose there are $t$ recursive calls (and thus pruning phases) between sampling $\overline{c_1}$ and $\overline{c_2}$.
Then, $N = N^{(1)} \cup  \dots \cup  N^{(t)}$ where $|N^{(i)}| = \frac{n}{2^i}$, since in each pruning phase $\frac{1}{2}|\overline{P}|$ points are removed.
Intuitively, each $N^{(i)}$ contains only a few points from $P_2$. Let $\overline{P}^{(0)} = P$ and $\overline{P}^{(i)} = \overline{P}^{(i - 1)}\setminus N^{(i)}$.
Recall that when $|P_2 \cap \overline{P}| \geq \alpha |\overline{P}|$, we assign these points to $\overline{c_2}$. Thus, $| P_2 \cap \overline{P}^{(i)} | < \alpha | \overline{P}^{(i)} |$ for all $i < t$. Thus, we bound the number of points of $P_2$ in each $N^{(i)}$ for $i \leq t$:

\begin{equation}
    \label{eq:p2-intersect-N}
|P_2 \cap N^{(i)} | \leq | P_2 \cap \overline{P}^{(i - 1)} | < \alpha |\overline{P}^{(i-1)}| = 2\alpha\frac{n}{2^i}.
\end{equation}

because $N^{(i)} \subset \overline{P}^{(i - 1)}$. We can also bound the number of points of $P_1$ in each $N^{(i)}$:

\begin{equation}
    \label{eq:p1-intersect-N}
|P_1 \cap N^{(i)} | \geq |N^{(i)}| - | P_2 \cap N^{(i)} | \geq  (1- 2\alpha)\frac{n}{2^i}.
\end{equation}

We first show that assigning $P_2 \cap N$ to $\overline{c_1}$ has small cost. If $p \in N^{(i)}$ and $p' \in N^{(i+1)}$, it must be that $\dist(p, \overline{c_1}) \leq \dist(p', \overline{c_1})$ since minimal points with respect to $\overline{C}$ are chosen at each step to be removed. Thus we can sum over such $p, p'$, and for all $i < t$

\[ \frac{1}{|P_2 \cap N^{(i)}|} \sum_{p \in P_2 \cap N^{(i)}} \dist(p, \overline{c_1}) \leq \frac{1}{|P_1 \cap N^{(i+1)}|} \sum_{p \in P_1 \cap N^{(i+1)}} \dist(p, \overline{c_1}). \]

Combining this with Equations~\ref{eq:p2-intersect-N} and~\ref{eq:p1-intersect-N}, for all $i < t$ we obtain

\begin{equation}
    \label{eq:bound-p2-c1-cost}
    \sum_{p \in P_2 \cap N^{(i)}} \dist(p, \overline{c_1}) \leq \frac{4\alpha}{1 -2\alpha}   \sum_{p \in P_1 \cap N^{(i+1)}} \dist(p, \overline{c_1}).
\end{equation}

Finally, we bound the cost of assigning points in $P_2 \cap N^{(t)}$ to $\overline{c_1}$. We have $|P_1 \cap \overline{P}^{(t)}| = |\overline{P}^{(t)}| - |P_2 \cap \overline{P}^{(t)}|  \geq  |\overline{P}^{(t)}| - |P_2 \cap \overline{P}^{(t-1)}| > (1 - 2\alpha) \frac{n}{2^t}$. We combine this with our previous observations:

\begin{equation}
    \label{eq:bound-p2-c1-cost-Nt}
    \sum_{p \in P_2 \cap N^{(t)}} \dist(p, \overline{c_1}) \leq \frac{2\alpha}{1 -2\alpha}   \sum_{p \in P_1 \cap \overline{P}^{(t)}} \dist(p, \overline{c_1}).
\end{equation}

Now, we have computed a bound for assigning points of $P_2$ to $\overline{c_1}$ in each $N^{(i)}$. Combining these:

\[ \begin{aligned}
\sum_{p \in P_2 \cap N} \dist(p, \overline{c_1}) &= \sum_{i = 1}^t \sum_{p \in P_2 \cap N^{(i)}} \dist(p, \overline{c_1}) \\
&\leq \frac{4\alpha}{1 - 2\alpha} \sum_{i = 1}^{t - 1} \sum_{p \in P_1 \cap N^{(i+1)}} \dist(p, \overline{c_1}) + \frac{2\alpha}{1 - 2\alpha}  \sum_{p \in P_1 \cap \overline{P}^{(t)}} \dist(p, \overline{c_1}) \\
&\leq 8 \alpha \sum_{i = 1}^{t-1} \sum_{p \in P_1 \cap N^{(i+1)}} \dist(p, \overline{c_1}) + 8\alpha \sum_{p \in P_1 \cap \overline{P}^{(t)}} \dist(p, \overline{c_1})  \\
&\leq 8\alpha \sum_{p \in P_1} \dist(p, \overline{c_1}) \\
&\leq 8\alpha(1+\eps) \sum_{p \in P_1} \dist(p, c^*_1),
\end{aligned} \]
for $\alpha \le \frac{1}{4}$.

Thus for Case 1, we can bound the algorithm cost from Equation~\ref{eq:alg-cost}:

\[ \begin{aligned}
    \sum_{p \in P} \dist(p, \overline{C}) &\leq \sum_{p \in P_1} \dist(p, \overline{c}_1) + \sum_{p \in P_2 \cap N} \dist(p, \overline{c_1}) + \sum_{p \in P_2 \cap \overline{P}} \dist(p, \overline{c_2}) \\
    &\leq (1+\eps) \sum_{p \in P_1} \dist(p, c^*_1) + 8\alpha(1+\eps) \sum_{p \in P_1} \dist(p, c^*_1) + (1+\eps) \sum_{p \in P_2} \dist(p, c^*_2) \\
    &\leq (1+8\alpha)(1+\eps) opt_2(P)
\end{aligned} \]

{\bf Case 2:} If there is no recursive call with $|P_2 \cap R| \geq \alpha |R|$ the pruning phase will be called recursively $\lceil \log n \rceil$ times until there is a single point $q \in \overline{P}$. $q$ can be assigned to a cluster by itself at a cost of $0$. Then, $N = P \setminus \{ q \}$ and the proof of the previous case also bounds the cost of assigning $P \cap N$ to $\overline{c_1}$.

For $k > 2$, consider the algorithm as each center is added. Let $\overline{C_i} = \{ \overline{c_1}, \dots \overline{c_i} \}$ be the $i$ medians already approximated, corresponding to supercluster $P'_1$, with $P'_2$ consisting of clusters whose medians are yet to be found. Similar analysis as above shows that the cost of points incorrectly assigned to centers in $\overline{C_i}$ in the pruning phases can be bounded by $8 \alpha k \sum_{p \in P'_1} \dist(p, \overline{C_i})$. See the proof of Theorem~2.5 of~\cite{ABS10} for full details.
\end{proof}

\subsection{Proof of Lemma~\ref{lem:cluster_runtime}}
\label{app:proof_cluster_runtime}
\begin{proof}
We will first count the number of operations required, where computing $\dist$ between points in $X$ and $C$, and computing the closest point in $C$ for any point in $X$ both count as one operation each.

Let $T(n, k)$ denote the number of operations required by \textsc{Cluster} with $n$ input points and $k$ medians to be found. For $k=0$, we clearly have $T(n,0)  = O(1)$. For $n \le k$, we can put each input point in its own cluster, and return the points in $C$ closest to each input point as the cluster medians. Thus, $T(n,k) \le O(n)$.

Let us consider the case $n > k \ge 1$. In the sampling phase, the number of candidate centers generated is $2^{O\left(m_{\delta,\eps} \log\left(\frac{1}{\alpha}m_{\delta,\eps}\right)\right)}\cdot w(m_{\delta,\eps})$, each taking $f(m_{\delta,\eps})$ operations. Each of the candidate centers is then tried recursively, each taking $T(n,k-1)$ operations. The pruning phase takes $O(n)$ operations. After pruning, the algorithm is called once for the remaining point set, requiring $T(n/2,k)$ operations. We thus have
\begin{align*}
T(n,k) \le& 2^{O\left(m_{\delta,\eps} \log\left(\frac{1}{\alpha}m_{\delta,\eps}\right)\right)} \cdot w(m_{\delta,\eps}) \cdot \left(f(m_{\delta,\eps}) + T(n,k-1)\right) + T(n/2,k) + O(n) \\
\le&  2^{O\left(m_{\delta,\eps} \log\left(\frac{1}{\alpha}m_{\delta,\eps}\right)\right)} \cdot w(m_{\delta,\eps}) \cdot f(m_{\delta,\eps}) \cdot T(n, k-1) + T(n/2,k) + O(n).
\end{align*}
Solving the recurrence yields
\begin{align*}
T(n,k) = n \cdot 2^{O\left(k m_{\delta,\eps} \log\left(\frac{1}{\alpha}m_{\delta,\eps}\right)\right)} \cdot\left( w(m_{\delta,\eps}) \cdot f(m_{\delta,\eps})\right)^{O(k)}.
\end{align*}
Taking into account the time to compute $\dist$ and the closest point in $C$, we get the total running time to be $n \cdot 2^{O\left(k m_{\delta,\eps} \log\left(\frac{1}{\alpha}m_{\delta,\eps}\right)\right)} \cdot \left(w(m_{\delta,\eps}) \cdot f(m_{\delta,\eps})\right)^{O(k)} \cdot (\Q(C) + \D(C))$.
\end{proof}

\section{Proof of Theorem~\ref{thm:cover_strong}}
\label{app:proof_cover_strong}
The following two useful lemmas hold for any metric space $\X = (X, \dist)$.

\begin{lemma}[\cite{ABS10}, Lemma 3.2]
\label{lem:sampling_1}
Let $c \in X$, $P \subseteq X$ of size $n$, and $\delta > 0$. A uniform sample multiset $S \subseteq P$ of size $m$ satisfies
\begin{align*}
\Pr \left[ \exists q \in S \mid \dist(q,c) \ge \frac{1}{\delta n} \sum_{p \in P} \dist(p,c) \right] \le m\delta.
\end{align*}
\end{lemma}

\begin{lemma}[\cite{ABS10}, Lemma 3.3]
\label{lem:sampling_2}
Let $\eps \in (0,1]$, $P \subseteq X$ of size $n$, and $b,c \in X$ be such that $\sum_{p \in P} \dist(p,b) > (1+ \tfrac{4}{5} \eps) \sum_{p \in P} \dist(p,c)$. A uniform sample multiset $S \subseteq P$ of size $m$ satisfies
\begin{align*}
\Pr \left[ \sum_{s \in S} \dist(s,b) \le \sum_{s \in S} \dist(s,c) + \frac{\eps m}{5n} \sum_{p \in P} \dist(p,c) \right] < \exp\left( - \frac{\eps^2 m}{144}\right).
\end{align*}
\end{lemma}

Intuitively, Lemma~\ref{lem:sampling_1} says that if we are sampling points uniformly, the probability of sampling a point far away from a fixed point increases with the size of the sample. Lemma~\ref{lem:sampling_2} states that if the average distance of a point set to a point $c$ is smaller than to a point $b$, then the average distance of a uniform sample of the point set to $c$ relative to $b$ is also small, with some probability that can be lower-bounded.

We now prove Theorem~\ref{thm:cover_strong}.

\begin{proof}[Proof of Theorem~\ref{thm:cover_strong}]
Let $P, S,\eps,\delta$ be as specified in Definition~\ref{def:strong_sampling}. Let $|P| = n$ and $|S| = m$; the value of $m$ will be set later in the proof.

Let $r = \frac{6m}{\delta n} \sum_{p \in P} \dist(p,c_P)$. Let $r' = \frac{r}{3}, U = \ball(c_P,r),$ and $U' = \ball(c_P,r')$. By Lemma~\ref{lem:sampling_1}, for all $s \in S$, $\dist(c_P,s) \le r'$ with probability at least $1-\delta/2$. Thus, $S \subseteq U'$ with probability at least $1-\delta/2$. If $c_S \notin U$, then $\sum_{s \in S} \dist(s,c_S) > 2r'm$ (by triangle inequality). However since $S \subseteq U'$, we have $\sum_{s \in S} \dist(s,c_P) \le r'm$, contradicting the claim that $c_S$ is the optimal $(1,C)$-median of S. Thus, with probability at least $1-\delta/2$, $c_S \in U$.

Since $C$ is $g$-coverable, there exists a $\left(\frac{\eps}{5n}\sum_{p \in P} \dist(p,c_P)\right)$-cover of $U \cap C$ of size $g\left(\frac{30m}{\delta \eps}\right)$. Let $C'$ be such a cover. Define
\begin{align*}
C'_{bad} = \{c \in C' \mid \sum_{p \in P} \dist(p,c) > (1+\tfrac{4}{5}\eps)\sum_{p \in P} \dist(p,c_P)\}.
\end{align*}
Setting $m$ to be the sufficiently large constant $m_{\delta, \eps,g}$ and using the union bound and Lemma~\ref{lem:sampling_2} we get
\footnote{We assume that $g$ is such that for large enough $m$,  $g\left(\frac{30m}{\delta \eps}\right) \exp\left(\frac{-\eps^2 m}{144}\right) < \delta/2$.}
\begin{align*}
&\Pr \left[ \exists c \in C'_{bad} \mid \sum_{s \in S} \dist(s,c) \le \sum_{s \in S}\dist(s,c_P) + \frac{\eps m}{5n} \sum_{p \in P} \dist(p,c_P) \right] \\
< & g\left(\frac{30m}{\delta \eps}\right) \exp\left(\frac{-\eps^2 m}{144}\right)\\
< & \delta/2.
\end{align*}
Thus, with probability at least $1-\delta/2$, for all $c \in C'_{bad}$ we have
\begin{align*}
\sum_{s \in S} \dist(s,c) > \sum_{s \in S}\dist(s,c_P) + \frac{\eps m}{5n} \sum_{p \in P} \dist(p,c_P).
\end{align*}
Let $c'$ be the closest point in $C'$ to $c_S$. By definition of $C'$ and the fact that  $c_S \in U \cup C$, we have $\dist(c_S, c') \le \frac{\eps}{5n}\sum_{p \in P}\dist(p,c_P)$. We then have
\begin{align*}
\sum_{s \in S} \dist(s,c') &\le \sum_{s \in S} \left( \dist(s,c_S) + \dist(c_S, c')\right) \le  \sum_{s \in S} \dist(s,c_S) + \frac{\eps m}{5n} \sum_{p \in P}\dist(p,c_P)\\
&\le \sum_{s \in S} \dist(s,c_P) + \frac{\eps m}{5n} \sum_{p \in P}\dist(p,c_P) < \sum_{s \in S} \dist(s,c)
\end{align*}
for all $c \in C'_{bad}$ from previous inequality. Thus $c' \notin C'_{bad}$ and hence
\begin{align*}
\sum_{p \in P} \dist(p,c') \le (1+\tfrac{4}{5}\eps)\sum_{p \in P} \dist(p,c_P).
\end{align*}
We then conclude
\begin{align*}
\sum_{p \in P} \dist(p,c_S) \le \sum_{p \in P} \dist(p,c') + n \dist(c', c_S) \le (1+\eps) \sum_{p \in P} \dist(p,c_P).
\end{align*}
This event holds with probability at least $(1-\delta/2)^2 > 1-\delta$.
\end{proof}

\end{document}